\documentclass[a4paper,11pt]{article} %{amsart}

\usepackage{amsmath, amsthm, amssymb, marvosym}

\usepackage{enumerate}

\usepackage{color}

\usepackage{graphicx}
\usepackage{epstopdf}

\usepackage{soul}

\usepackage[normalem]{ulem}

\usepackage{subcaption}

\newcommand{\caF}{{\mathcal F}}

\newcommand{\caU}{{\mathcal U}}

\newcommand{\caZ}{{\mathcal Z}}

\newcommand{\bbC}{{\mathbb C}}

\newcommand{\bbN}{{\mathbb N}}

\newcommand{\bbR}{{\mathbb R}}

\newcommand{\ie}{{\it i.e.\/} }

\newcommand{\iu}{\mathrm{i}}
\newcommand{\str}{^*}
\newcommand{\Tr}{\operatorname{Tr}}

\newcommand{\ep}[1]{\mathrm{e}^{#1}}
\newcommand{\dif}{\mathrm{d}}
\newcommand{\Idif}{\,\mathrm{d}}

\newcommand{\ket}[1]{| #1 \rangle}
\newcommand{\bra}[1]{\langle #1 |}
\newcommand{\tr}{\operatorname{tr}}
\renewcommand{\Re}{\operatorname{Re}}

\renewcommand{\d}{\mathrm{d}}

 \newtheorem{thm}{Theorem}
 
 \newtheorem{cor}[thm]{Corollary}
 \newtheorem{lemma}[thm]{Lemma}
 \newtheorem{prop}[thm]{Proposition}

\newcommand{\comment}[1]{}

\newcommand{\rate}{\alpha}
\newcommand{\rtn}{\beta}

\newcommand{\No}{\caZ}

\author{Sven Bachmann
\\
\small{Mathematisches Institut der Universit{\"a}t M{\"u}nchen, 80333 M{\"u}nchen, Germany}
\\
Martin Fraas
\\
\small{980 Ohlone Ave. 991, Albany, CA 94706, USA}
\\ Gian Michele Graf
\\
\small{Theoretische Physik, ETH Z\"urich, 8093 Z\"urich, Switzerland} }

\begin{document}

\title{Dynamical crossing of an infinitely degenerate critical point}
\date{\today}%
\maketitle

\begin{abstract}
We study the evolution of a driven harmonic oscillator with a time-dependent frequency $\omega_t \propto |t|$. At time $t=0$ the Hamiltonian undergoes a point of infinite spectral degeneracy. If the system is initialized in the instantaneous vacuum in the distant past then the asymptotic future state is a squeezed state whose parameters are explicitly determined. We show that the squeezing is independent on the sweeping rate. This manifests the failure of the adiabatic approximation at points where infinitely many eigenvalues collide. We extend our analysis to the situation where the gap at $t=0$ remains finite. We also discuss the natural geometry of the manifold of squeezed states. We show that it is realized by the Poincar\'e disk model viewed as a K\"ahler manifold.
\end{abstract} 

\section{Introduction}

In a generic situation, the energy levels of a parameter-dependent Hamiltonian may get very close to each other but do not cross. Such an `avoided crossing' generically has a hyperbolic shape, and the Landau-Zener Hamiltonian is a paradigmatic model of a driven dynamics undergoing such a crossing. On the other hand, degenerate energy points may occur because of symmetry reasons. This is typically the case at quantum phase transitions, where the energy of a finite or an infinite number of low-lying excited states equals the ground state energy at the critical point. A theoretically well-understood and experimentally relevant model of this phenomenon is the Dicke model~\cite{Dicke}.

The Dicke model describes the interaction of a large number of two level atoms coupled to a single photonic mode. It exhibits a quantum phase transition at a critical coupling strength. The ground state below the critical coupling has no photons and all atoms are in their respective ground states, while above the critical coupling there is a macroscopic amount of excitations of both field and atoms in the ground state, see~\cite{HeppLieb} for a complete picture in the rotating wave approximation. In between a macroscopic number of eigenmodes all `collapse' onto the ground state.

Let the system be prepared in its ground state at zero coupling. The coupling is then ramped at a sweeping rate $\rate$ across the critical point. If the adiabatic approximation were to hold, the system should be at all times in a state that is close to its instantaneous ground state as the rate goes to zero. Discrepancies to this rule indicate diabatic transitions. They are observed in experiments~\cite{Esslinger}.

The subject of this letter is an exactly solvable model of such a dynamical transition across an infinitely degenerate critical point. We show that the system is far from its ground state after the crossing and we in fact compute explicitly the asymptotic distribution in the energy levels. The `symmetric' shape of the crossing in the model described below is non-generic, and although it was motivated by it, it is not meant to predict the exact distribution in the Dicke model. Indeed, the dependence of the non-adiabatic error on the sweeping rate is not universal, the scaling of the error being predicted by the Kibble-Zurek theory \cite{Zurek}. Finally, we note that for finitely degenerate points the adiabatic approximation does hold, but with a shape dependent rate~\cite{Hagedorn}.

\section{The result}

Concretely, we consider the time-dependent harmonic oscillator with a frequency $\omega_t$ and Hamiltonian given by
\begin{equation}\label{Hamiltonian}
H_t = \frac{1}{2}\bigl(p^2 + \omega_t^2 x^2\bigr)
\end{equation}
on $L^2(\bbR_x)$. Whenever $\omega_t=0$, the discrete spectrum $\{(n+\frac{1}{2})\omega_t: n\in\bbN\}$ collapses to the  purely absolutely continuous spectrum $[0,\infty)$. We shall study the dynamics generated by this Hamiltonian with frequency $\omega_t = \rate |t|$ characterized by a constant angular acceleration $\rate>0$, see Figure~\ref{fig:levels}.

\begin{figure}[htb]
\centering
\includegraphics{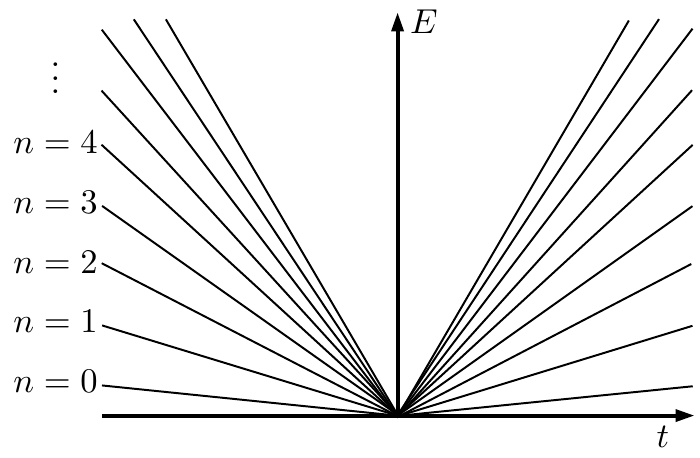}
\caption{The collapse of the spectrum of $H_t$. We plot the first few eigenvalues as a function of the parameter $t$.}
\label{fig:levels}
\end{figure}

It is convenient to introduce the annihilation operator
\begin{equation}\label{Classical a}
a = \frac{1}{\sqrt{2\omega}}(\omega x +\iu p),
\end{equation}
for some fixed $\omega>0$ independent of time. Eventually, we will be interested in the evolution over a symmetric time interval $[-t,t]$ and take $\omega = \omega_t = \omega_{-t}$. Then $a\psi_0=0$ characterizes the vacuum $\psi_0$ of both Hamiltonians $H_{-t}=H_t$.

A {\it squeezed state} $\psi$ is a state satisfying 
\begin{equation}\label{SqSt}
(\lambda a+\mu a^*)\psi=0
\end{equation}
for some $\lambda, \mu\in\bbC$ with $(\lambda, \mu)\neq(0,0)$. The
state $\psi$ exists iff $|\mu|<|\lambda|$, in which case it is
determined up to a phase by $\lambda, \mu$. In fact, the definition is
restated by (\ref{Classical a}) as the ordinary differential equation
$(\lambda-\mu)\psi'(x)=-\omega (\lambda+\mu)x\psi(x)$.
Except for $\lambda=\mu$ it has a solution, which is
\begin{equation*}
\psi(x)=\ep{-\omega\frac{\lambda+\mu}{\lambda-\mu}\frac{x^2}{2}}
\end{equation*}
up to multiples; it is in $L^2(\bbR_x)$ iff the quotient has
positive real part. The latter is equivalent to the stated condition $|\mu|<|\lambda|$.

Given that $|\mu/\lambda|<1$, we can write
\begin{equation}\label{quot}
\frac{\mu}{\lambda}=\ep{\iu\theta}\tanh r
\end{equation}
for some $r\ge 0$, $\theta\in [0,2\pi)$. We conclude that squeezed states are parameterized by a complex number $\tau=r\ep{\iu\theta}$. The ground state, $a\psi_0=0$, corresponds to $\tau=0$. Other eigenstates $\psi_n, (n=1,\ldots)$ of the excitation number operator $a^*a$ are not squeezed states.

Our main result is the following description of the evolution.
\begin{thm}\label{THM}
Let the system be in its ground state at time $-t$; then at time $t$ it is in a squeezed state $\psi$ of parameter $r\ep{\iu \theta}$ with
\begin{equation*}
\theta = -\rate t^2 - \frac{\pi}{2}+ o(1),\qquad \tanh r = \frac{1}{\sqrt 2} + o(1),
\end{equation*}
as $t\to\infty$. In particular the probability of producing $n$ excitations, 
$p_n=|\langle \psi_n, \psi\rangle |^2$, ($n=0,1,\ldots$), is
\begin{equation*}
p_{2k}(+\infty)=\frac{1}{2^k\sqrt{2}}\cdot\frac{(2k-1)!!}{(2k)!!},\qquad
p_{2k+1}(+\infty)=0.
\end{equation*}
\end{thm}
The squeezing is meant with respect to the ground state $\psi_0$ of $H_t = H_{-t}$. The probability of returning to it, or fidelity, is $p_0=1/\sqrt{2}$. \\
 
We shall present the main ideas of the proof shortly, leaving details to later sections. Before doing so we make three remarks. The second one, which is about self-similarity, is formulated as a lemma, since it will be used in the proof of the theorem. The other two are heuristic.\\

\noindent\textbf{Remark.} The evolution is manifestly non-adiabatic near $t=0$, where all gaps close. A general but rough criterion for the adiabatic regime is $\|dP/dt\|\ll \Delta$, where $P$ is the projection onto the eigenstate under consideration and $\Delta$ is the gap separating it from the remaining ones. Clearly $\Delta=\rate t$ here, whereas the ground state $P$ depends on $t$ through a dilation $x\mapsto\gamma x$ by $\gamma\propto(\alpha |t|)^\beta$; indeed the same applies to the Hamiltonian (\ref{Hamiltonian}), up to an overall factor which however does not affect $P$. There is no need to compute the exponent $\beta=-1/2$. In fact dilations form a one-parameter group if parameterized additively. Thus $\|dP/dt\|=O(d\log\gamma/dt)=O(t^{-1})$. The adiabatic criterion so amounts to $\rate t^2\gg 1$.

\begin{lemma}\label{self}
The special case $\rate =1$ suffices.
\end{lemma}

\begin{proof} The dynamics
$$
\iu\frac{\dif\psi(t)}{\dif t} = \frac{1}{2}\bigl(p^2 + \rate^2t^2 x^2\bigr) \psi(t),\qquad (\rate > 0) ,
$$
can be rescaled in space and time, $x = \lambda x'$, $t = \mu t'$. Then for $\mu^{-1} = \lambda^{-2}$, $\mu^{-1} = \rate^2\mu^2\lambda^2$, \ie for $\lambda = \mu^{1/2}, \mu = \vert \rate \vert^{-1/2}$, the state $\psi'(t',x'):= \psi(t,x)$ solves the Schr\"odinger equation for the Hamiltonian~(\ref{Hamiltonian}) with $\rate =1$. In particular the phase slip $-\pi/2$, the squeezing $r$, and the probabilities are independent of $\rate$, asymptotically in $t\to \infty$. 
\end{proof}

\noindent\textbf{Remark.} The contribution $-t^2$ to $\theta$, which reflects the WKB approximation, can be understood as follows. During early times $t\ll -t_0$, ($t_0\approx 1$) the dynamics is adiabatic, meaning that the evolved state closely shadows the instantaneous ground state. In particular there is no substantial squeezing of the former with respect to the latter. In between $-t_0$ and $+t_0$ the state makes a non-adiabatic transition to a squeezed state. Thereafter the axis of squeezing of that state rotates clockwise by an angle
\begin{equation*}
-\rtn=\int_0^t\omega_{t'}\Idif t' = \frac{t^2}{2}.
\end{equation*}
Let $\mathcal{U}_\rtn$ be the unitary transformation corresponding to a rotation in phase space by $\rtn$, whence $\mathcal{U}_\rtn^*a\mathcal{U}_\rtn=\ep{\iu\rtn}a$. For a squeezed state $\psi$ of parameter $\mu/\lambda = e^{i\theta}\tanh r $, the state $\psi_\rtn=\mathcal{U}_\rtn\psi$ then satisfies $(\lambda\ep{-\iu\rtn} a+\mu\ep{\iu\rtn} a^*)\psi_\rtn=0$ and is thus a squeezed state of parameter $\ep{\iu(\theta+2\rtn)}\tanh r$, in agreement with the said contribution. \\

\begin{proof}[Proof of Theorem~\ref{THM}] We interpret $(x,p)\in\mathbb{R}^2$ as phase space coordinates. Complex coordinates $a$ and $\bar a$ are defined by (\ref{Classical a}) and its complex conjugate. A linear symplectic map is then represented as a matrix
\begin{equation}
\label{ClassicalFlow}
\varphi:= \begin{pmatrix}
U & \overline{V} \\ V & \overline U
\end{pmatrix} \in \mathrm{SU}(1,1) , \qquad \begin{pmatrix}	
a  \\ \overline{a}\end{pmatrix} \mapsto \varphi \begin{pmatrix}
									a \\ \overline{a}
										\end{pmatrix}.
\end{equation}
A matrix of that form belongs to the group $\mathrm{SU}(1,1)$ if the condition $|U|^2 - |V|^2 =1$ holds true. The form and the condition state that the map is compatible with complex conjugation, respectively that it leaves invariant
the sesquilinear form  $\bar a_2a_1-a_2 \bar a_1=\iu(x_2p_1-p_2x_1)$ associated to the phase space volume.

The classical equations of motion associated to (\ref{Hamiltonian}) 
\begin{equation}
\label{ClassEq}
\dot x(t) = p(t),\qquad \dot p(t) = -t^2 x(t),
\end{equation}
generate a linear Hamiltonian flow, which for any $t$ is an example for a $\varphi$ as in (\ref{ClassicalFlow}). As we shall see at the end of Sect.~\ref{sec:flow} the propagator for the interval $[-t,t]$  is given by
\begin{equation}\label{ClassicalSol}
\begin{aligned}
U &= - \sqrt{2} \iu \ep{-\iu (t^2 - \frac{\pi}{2})}+ o(1), \\
\overline{V} &= \iu + o(1),
\end{aligned}
\end{equation}
as $t\to+\infty$. 

The quantum evolution for the same interval is given by the propagator $\mathcal{U}$ in the Schr\"odinger picture; then in the Heisenberg picture by $A\mapsto \caU\str A \caU$. Since the Heisenberg equations of motion are formally identical with the canonical equations of motion (up to $a\str$ replacing $\bar a$) and moreover linear, the time evolution of $a$ and $a^*$ is given by the Bogoliubov transformation (\ref{ClassicalFlow}), \ie
\begin{equation*}
\mathcal{U} \str a \mathcal{U} = U a + \overline{V} a^*,
\end{equation*}
and its hermitian conjugate. The classical flow thus completely determines the quantum evolution. 

The final state $\caU\psi_0$ satisfies
\begin{equation*}
\caU\str(\lambda a+\mu a^*)\caU\psi_0=(\lambda\overline{V}+\mu\overline{U}) a^*\psi_0
\end{equation*}
by $a\psi_0=0$. It thus is a squeezed state for parameters (\ref{quot}) making the r.h.s. vanish:
\begin{equation}\label{Csqst}
\ep{\iu\theta}\tanh r=\frac{\mu}{\lambda}= -\frac{\overline V}{\overline U}. 
\end{equation}
Comparison with (\ref{ClassicalSol}) concludes the proof of the first claim. The transition probabilities will be computed following Lemma~\ref{lemma:U squeezed}.
\end{proof}

The squeezed state (\ref{SqSt}) is conveniently parameterized by the complex coordinate $z = \ep{\iu \theta} \tanh r$ seen in (\ref{quot}) and denoted $\ket{\psi(z)}$. The family of squeezed states endows the disk $\vert z \vert <1$ with natural metric and curvature tensors (see Section~\ref{sec:geom}).
\begin{prop}\label{sqpr}
Let $P_z = \ket{\psi(z)} \bra{\psi(z)}$ be a family of squeezed projections. Then the associated Fubini-Study metric and adiabatic curvature are given by
$$
g = \frac{1}{(1 -|z|^2)^2} |\d z|^2, \qquad \omega = \frac{1}{(1- |z|^2)^2}  \frac{\iu}{2} \d z\wedge \d \bar{z} .
$$
The unit complex disk endowed with $g$, $\omega$ realizes the Poincar\'{e} disk model; in particular it is a K\"{a}hler manifold.
\end{prop}

The proposition is corollary of Proposition~\ref{Ngeom} which describes the metric and curvature tensors associated to $N$-mode squeezed states. The metric and curvature tensors obtained are the standard tensors associated to the Poincar\'{e} disk model, up to a constant prefactor.

The time dependent squeezed state described in Theorem~\ref{THM} corresponds to a trajectory on the Poincar\'{e} disk. We plot this trajectory in Figure~\ref{fig:SqueezingPlot}. The trajectory hits the boundary of the disk at the critical time $t=0$ and then spirals into its final point described in Theorem~\ref{THM}.

\begin{figure}[htb]
\centering
\begin{subfigure}{.55\textwidth}
\includegraphics[width = 0.9\textwidth]{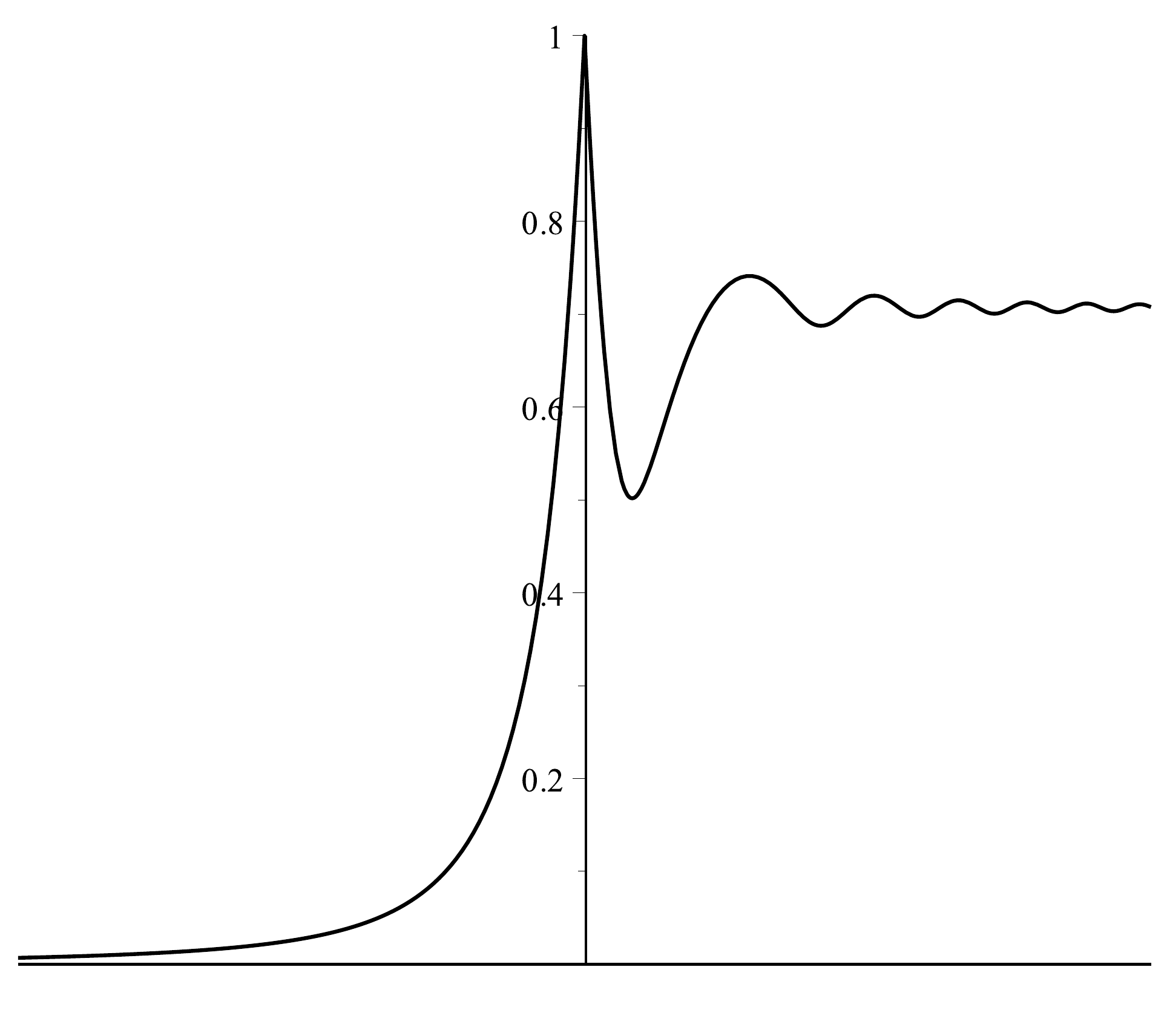}
\caption{}
\end{subfigure}%
\begin{subfigure}{.45\textwidth}
\includegraphics[width = 0.9\textwidth]{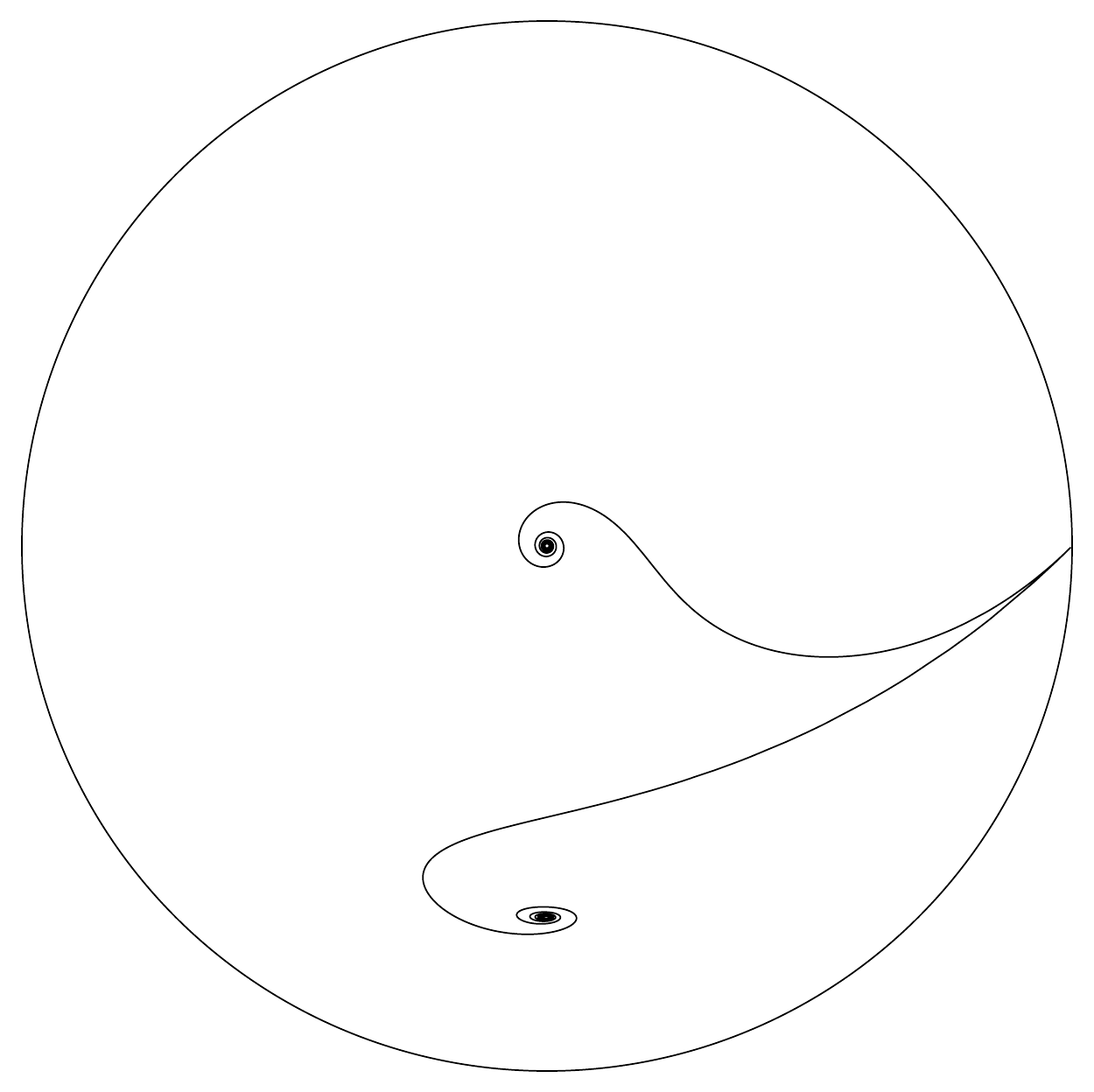}
\caption{}
%\label{fig:DiskPlot}
\end{subfigure}
\caption{We plot the squeezing parameter of the evolved state considered in Theorem~\ref{THM} as a function of time $t$ and with respect to the instantaneous ground state. (a) The modulus $\tanh r$, illustrating the breakdown of adiabaticity near the collapse $t=0$, and the transition from no squeezing when $t \ll -1$ to a squeezing of $\tanh r = 1/\sqrt{2}$ for $t \gg 1$. (b) The trajectory in the Poincar\'{e} disk after subtracting the $-t^2$ contribution to the phase, \ie we plot $w:=\ep{\iu t^2} z$ with asymptotic values $w=0$ and $w=-\iu/\sqrt{2}$.}
\label{fig:SqueezingPlot}
\end{figure}

Before deriving the asymptotics~(\ref{ClassicalSol}) of the solution of the classical flow, we describe in the next section the general mathematical structure of driven quadratic systems, thus providing the framework for the theorem above. We recall in particular Bogoliubov transformations and quasi-free states for the case of $N$ modes. We also discuss squeezed states in more details and derive their occupation numbers. Section~\ref{sec:flow} then provides the only missing piece in the proof of Theorem~\ref{THM}, namely the solution of the classical flow. In Section~\ref{sec:Extensions} we discuss two extensions of Theorem~\ref{THM}. In the first extension we consider non-symmetric time intervals, while in the second one we consider a two parameter frequency profile $\omega_t = \sqrt{(\alpha t)^2 + g^2}$ which includes an avoided crossing. We conclude with a summary in Section~\ref{sec:Conclusions}.

%%%%%%%%%%%%%%%%%%%%%%%%%%%%%%%%%%%%%%%%%%%%%%%%%%%%%%%%%%%%%%%%%%%%%%%%

%%%%%%%%%%%%%%%%%%%%%%%%%%%%%%%%%%

\section{Bogoliubov transformations and squeezed states}
\label{sec:Bog}

The driven one-dimensional harmonic oscillator is a special case of a qua\-dra\-tic Hamiltonian with time-dependent coefficients. The Heisenberg evolution still reduces to a classical dynamics, namely to a family of Bogoliubov transformations. When viewed in the Schr\"odinger picture, squeezed states are mapped to just such.
\subsection{An algebraic view on squeezed states}
\label{sec:Alg}
Bogoliubov transformations are matrices
\begin{equation}\label{genBog}
\varphi := \begin{pmatrix}
U & \overline{V} \\ V & \overline U
\end{pmatrix} \in \mathrm{SU}(N,N)
\end{equation}
that act as automorphisms on the CCR-algebra generated by $a_i,a_i\str$, $(i=1,\ldots, N)$ through
\begin{equation}\label{Classical&QuantumBog}
a_i\mapsto \sum_{j=1}^n U_{ij} a_j + \overline{V_{ij}} a_j\str
\end{equation}
and the hermitian conjugate thereof. By introducing
\begin{equation*}
a(f):= \sum_{i=1}^N a_i \overline{f_i},\qquad (f\in\bbC^N)
\end{equation*}
it equivalently acts as
\begin{equation}\label{aBog}
a(f)\mapsto a(U\str f) + a\str(V\str\overline f)
\end{equation}
and its conjugate. In terms of the self-dual annihilation operators 
\begin{equation}\label{selfdualA}
A(\xi) := a(f) + a\str(Cg), \qquad (\xi = f\oplus g)
\end{equation}
with $C:g\mapsto\overline{g}$ denoting complex conjugation, it reads
\begin{equation}
\label{selfdualBog}
A(\xi)\mapsto A(\varphi\str\xi),
\end{equation}
where
\begin{equation}
\label{stared}
\varphi\str = \begin{pmatrix}
U\str & V\str \\ V^t & U^t
\end{pmatrix}.
\end{equation}
In fact, eq.~(\ref{selfdualBog}) states for $\xi = f\oplus g$ that
\begin{equation*}
a(f) + a\str(Cg) \mapsto a(U\str f + V\str g) + a\str(C(V^t f + U^t g))
\end{equation*}
which for $g=0$ agrees with~(\ref{aBog}) and for $f=0$ with the conjugate thereof.

It is convenient to introduce the matrices 
\begin{equation*}%\label{js}
J = \begin{pmatrix} 0 & C \\ C & 0 \end{pmatrix},\qquad S = \begin{pmatrix} 1 & 0 \\ 0 & -1\end{pmatrix}
\end{equation*}
satisfying $JS+SJ=0$. The properties of the operators $A(\xi)$ are then stated as
\begin{align}
&\xi\mapsto A(\xi)\text{ is antilinear},\nonumber \\
&[A(\xi),A\str(\zeta)] = \langle \xi, S\zeta\rangle,\label{cr} \\
&A(J\xi) = A\str(\xi),\nonumber
\end{align}
and a \emph{Bogoliubov transformation} is an invertible map $\varphi$ on $\bbC^N\oplus\bbC^N$ such that
\begin{equation*}
\varphi S \varphi\str = S, \qquad J \varphi\str = \varphi\str J.
\end{equation*}
To avoid any confusion we remark that all scalar products are denoted $\langle \cdot, \cdot \rangle$, regardless of the vector space (e.g. $\bbC^N$ or $\bbC^N\oplus\bbC^N$), and that they are linear in the second factor.

\emph{Quasi-free states} are states on the self-dual algebra specified by 
\begin{equation}
\label{eq:sd-state}
\omega_P(A(\xi)A\str(\zeta)) = \langle \xi, S P\zeta\rangle,
\end{equation}
and by Wick's rule, where the matrix $P$ is such that
\begin{align}
& \langle P \xi, S\zeta\rangle = \langle \xi, SP\zeta\rangle,\label{eqs1}\\
&\langle \zeta, S P\zeta\rangle \ge 0,\label{eqs2} \\
&P + JPJ = 1.\label{eqs3}
\end{align}
Indeed, the conditions reflect $\omega_P(A\str)=\overline{\omega_P(A)}$, $\omega_P(AA\str)\ge 0$, and (\ref{cr}). The first one states $P\str S=SP$ and is a prerequisite for the second by $(SP)\str=SP$. The inequality is strict if $P\zeta\neq 0$. The last condition also reads $PJ=J(1-P)$.

In this setting, a \emph{squeezed state} is a quasi-free state which is indecomposable among general states. 

Gauge invariant quasi-free states are defined on the usual CCR-algebra by
\begin{equation}\label{GaugeQFS0}
\omega_\rho(a(g) a(f)) =0, \qquad\omega_\rho(a\str(g) a(f)) = \langle f,\rho g\rangle
\end{equation}
for some matrix $\rho=\rho^*$ and correspond to the special case 
\begin{equation}\label{GaugeQFS}
P_\rho = \begin{pmatrix} 1+\rho & 0 \\ 0 & -C\rho C \end{pmatrix}
\end{equation}
with $\rho\geq 0$ by (\ref{eqs2}); conditions (\ref{eqs1}, \ref{eqs3}) are satisfied identically. In particular, $P_\rho$ is a projection, $P_\rho=P_\rho^2$, iff $\rho = 0$.

A Bogoliubov transformation $\varphi$ induces one on arbitrary states and in particular on quasi-free states, $\omega_P\mapsto\omega_{\tilde P}$ by
\begin{equation}\label{omegaBog}
 \omega_{\tilde P}(A(\xi)A\str(\zeta))=
\omega_P(A(\varphi\str\xi)A\str(\varphi\str\zeta)), 
\end{equation}
\ie $P\mapsto \tilde P$ with
\begin{equation}\label{PBog}
\tilde P = (\varphi\str)^{-1}P \varphi\str,
\end{equation}
because (\ref{omegaBog}) equals
\begin{equation*}%\label{PBogCalc}
\langle \varphi\str\xi, SP\varphi\str\zeta\rangle = \langle \xi, \varphi SP\varphi\str\zeta\rangle =\langle \xi, S(\varphi\str)^{-1 }P\varphi\str\zeta\rangle.
\end{equation*}
We note by the way that including the `metric tensor' $S$ in the definition (\ref{eq:sd-state}) of quasi-free states is the reason that Bogoliubov transformations $\varphi$ act as similarity transformations (\ref{PBog}) on $P$. For instance $P$ is a projection iff $\tilde P$ is.

An application to dynamics is as follows. Let $\mathcal{U}:\mathcal{F}\mapsto\mathcal{F}$ be a propagator on Fock space $\mathcal F$ in the Schr\"odinger picture generated by a possibly time-dependent Hamiltonian that is quadratic in $\{a_i,a_i\str:i=1,\ldots,N\}$. Its classical counterpart, which is a function of the complex variables $a_i$, $\overline{a_i}$, generates a propagator $\varphi\in \mathrm{SU}(N,N)$ in the sense of~(\ref{Classical&QuantumBog}). Again, since the Heisenberg equations of motion are formally identical with the canonical equations of motion and moreover linear, we have
\begin{equation*}
A_\mathcal{U}(\xi):= \mathcal{U}\str A(\xi) \mathcal{U} = A(\varphi\str\xi),
\end{equation*}
see~(\ref{selfdualBog}). The expectations in the (initial) quasi-free state $P$ are thus
\begin{equation*}
%\label{StateEvolution}
\omega_P(A_\mathcal{U}(\xi)A_\mathcal{U}\str(\zeta)) = \omega_{\tilde P}(A(\xi)A\str(\zeta))
\end{equation*}
with $\tilde P$ as in (\ref{PBog}). In particular $P\mapsto \tilde P$ is the propagator on quasi-free states in the Schr\"odinger picture; it does not rely on the Fock space representation of the self-dual algebra.

We now provide two further applications which are of independent interest. Firstly, we prove that any quasi-free state is gauge-invariant up to a suitable Bogoliubov transformation. Secondly, we show that a quasi-free state $\omega_P$ is a squeezed state if and only if $P=P^2$.
\begin{lemma}\label{lemma:diagonalization}
Let $P$ satisfy conditions~(\ref{eqs1}-\ref{eqs3}). Then there exists a Bogoliubov transformation $\varphi$ such that $\varphi\str P (\varphi\str)^{-1} $ has the form (\ref{GaugeQFS}) with $\rho$ diagonal.
\end{lemma}
\begin{proof}
Let $\widetilde{P} := P-\frac{1}{2}$, in terms of which (\ref{eqs1}, \ref{eqs3}) hold true alike, except for $1$ replaced by $0$. The positivity of $SP$ and the identity $2 S\widetilde{P} = SP - S JPJ = SP + JSP J$ imply that $S\widetilde{P} \geq 0$. Actually, $S\widetilde{P} > 0$. Indeed, $\langle \zeta, S \tilde{P} \zeta\rangle = 0$ implies by the above identity $P\zeta=0$ and $PJ\zeta=0$, which by (\ref{eqs3}) gives $\zeta=0$.
Let now $\chi$ be an eigenvector of the self-adjoint $(S\widetilde{P})^{1/2}S(S\widetilde{P})^{1/2}$; then the non-zero vector $v := S(S\widetilde{P})^{1/2}\chi$ is an eigenvector of $S(S\widetilde{P}) = \widetilde{P}$ for the same eigenvalue $\lambda$. With this, $0< \langle v, S\widetilde{P} v\rangle = \lambda\langle v, S v\rangle$ and since $S$ is self-adjoint, both $\lambda$ and $\langle v, S v\rangle$ are real and non-zero. In particular, the quadratic form $S$ is positive definite on each eigenspace with $\lambda>0$, which thus has a basis $(v_i)$ with  $\langle v_i, S v_j\rangle = \delta_{ij}$. We further note that by $J \widetilde{P} J = -\widetilde{P}$ the vector $\hat v:= Jv$ is an eigenvector with opposite eigenvalue: $\widetilde{P} \hat v = -\overline \lambda \hat v=-\lambda \hat v$. Collecting the vectors $v_i$ for all positive eigenvalues, together with $\hat v_i=Jv_i$, we end up with a basis $v_1,\ldots,v_N,\hat v_1,\ldots \hat v_N\in\bbC^N\oplus\bbC^N$ satisfying 
\begin{equation*}
\langle v_i, Sv_j\rangle =\delta_{ij} ,\qquad \langle \hat v_i, S\hat v_i\rangle = -\delta_{ij},
\qquad\langle \hat v_i, Sv_j\rangle =0
\end{equation*}
because for eigenvectors $v$, $v'$ with eigenvalues $\lambda$, $\lambda'$ we have $(\lambda-\lambda')\langle v',Sv\rangle=0$ by (\ref{eqs1}).

Using the canonical basis $e_1,\ldots ,e_N$ of $\bbC^N$, the map $\varphi$ defined by
\begin{equation*}
\varphi\str v_i = e_i\oplus 0,\qquad \varphi\str \hat v_i = 0\oplus e_i,
\end{equation*}
is a Bogoliubov transformation. Indeed, $\varphi S \varphi\str = S$ since both sides share the same matrix elements, and $\varphi\str J = J\varphi\str$ by construction. Finally,
\begin{equation*}
\varphi\str\tilde P(\varphi\str)^{-1} (e_i\oplus 0) = 
\varphi\str\tilde Pv_i= \lambda_i  \varphi\str v_i  = \lambda_i (e_i\oplus 0)
\end{equation*}
and similarly $\varphi\str \tilde P(\varphi\str)^{-1} (0 \oplus e_i) = -\lambda_i (0\oplus e_i)$. Summarizing, for any $P$ there is a Bogoliubov transformation $\varphi$ such that $\varphi\str P (\varphi\str)^{-1} $ is of the form~(\ref{GaugeQFS}) with a diagonal matrix $\rho=\mathrm{diag}(\rho_1\ldots,\rho_N)$ obtained from the above eigenvalues by $\rho_i=\lambda_i - 1/2$.
\end{proof}
%%%%%%%%%%%%%%%%%%%%%%%%%
The Fock space $\mathcal{F}$ is a Hilbert space carrying an irreducible representation of the (finitely generated) CCR-algebra, and it is unique up to isomorphism. There is a vector, the Fock vacuum $\psi_0\in\caF$, ($\|\psi_0\|=1$) unique up to a phase and characterized by $a(f)\psi_0=0$, ($f\in \bbC^N$). We recall the functor of second quantization, $\Gamma$, which promotes single-particle operators $B$ on $\mathbb C^N$ to operators on $\mathcal{F}$ by $\Gamma(B)a\str(f)=a\str(Bf)\Gamma(B)$.

\comment{The standard representation of the Fock space is as the symmetric algebra $\mathcal{F} = \bigoplus_{k=0}^\infty(\mathbb{C}^N)^{\otimes_s^k}$. We recall the functor of second quantization, $\Gamma$, which promotes single-particle operators $B$ on $\mathbb C^N$ to operators on $\mathcal{F}$ by $\Gamma(B)=B\otimes \ldots\otimes B$ ($k$ times) on $k$-particle states in $\mathbb C^N\otimes \ldots\otimes\mathbb C^N$ for any $k\in\mathbb N$.}
\begin{lemma}\label{lemma:normal1}
Any gauge-invariant state $\omega_\rho$ is realized by a unique density matrix on $\mathcal{F}$, meaning $\omega_\rho(A)=\mathrm{Tr}_\caF(\nu_\rho A)$. In fact
\begin{equation}\label{density matrix}
\nu_\rho = Z^{-1}\Gamma(Q),\qquad Q=\rho(1+\rho)^{-1}
\end{equation}
with $Z=\Tr_{\mathcal F}\Gamma(Q)$. In particular $\nu_\rho$ is given by a vector iff $\rho=0$, in which case it is the Fock vacuum,
\begin{equation}\label{Fock vacuum}
\omega_0(A)=\langle\psi_0,A\psi_0\rangle.
\end{equation}
\end{lemma}
\begin{proof} Uniqueness of $\nu_\rho$ follows because the CCR-algebra is irreducibly represented on $\caF$.
 Existence: Following \cite{Gaudin:1960aa} we shall show that $\omega$, as defined by (\ref{density matrix}), is a quasi-free state with the appropriate two-point functions (\ref{GaugeQFS0}), namely
\begin{equation}\label{GaugeQFSbis}
\omega(a_i a_j) =0, \qquad\omega(a_i\str a_j)=\rho_i\delta_{ij},
\end{equation}
where we assumed without loss that $\rho$ is diagonal. Interchanging $a_k$ and $a_k\str$ in (\ref{GaugeQFSbis}) amounts to replace $\rho_i$ by $1+\rho_i$. Let $b_i$ be either $a_k$ or $a_k\str$ and set $[b_ib_j]:=b_ib_j - b_j b_i \in\{0,\pm1\}$. Then all four equations are summarized by 
\begin{equation}\label{GaugeQFSter}
\omega(a_i\str b_j)=-\rho_i[a_i\str b_j], \qquad
\omega(a_i b_j)=(1+\rho_i)[a_i b_j].
\end{equation}
Wick's rule will follow immediately by iteration from the claim 
\begin{equation}\label{wr}
\omega(b_0B)=\sum_{j=1}^n\omega(b_0b_j)\omega(B_j),
\end{equation}
where we set $B=b_1\cdots b_n$ and obtained $B_j$ by omitting from it the factor $b_j$. In fact by using $b_0b_j=b_jb_0+[b_0b_j]$ repeatedly we have $b_0B=Bb_0+\sum_{j=1}[b_0b_j]B_j$ and
\begin{equation}\label{wr1}
\Tr(\nu_\rho b_0B)=\Tr(\nu_\rho Bb_0)+\sum_{j=1}^n[b_0b_j]\Tr(\nu_\rho B_j).
\end{equation}
Moreover we have the pair
\begin{equation}\label{Qa}
a_i\Gamma(Q)=q_i\Gamma(Q)a_i,\qquad\Gamma(Q)a_i\str=q_ia_i\str\Gamma(Q). 
\end{equation}
of adjoint equations, of which one or the other will be used for $b_0=a_i$ and for $b_0=a_i\str$. In the first case we first use (\ref{wr1}) and then (\ref{Qa}) through $\Tr(\nu_\rho Ba_i)=\Tr(a_i\nu_\rho B)=q_i\Tr(\nu_\rho a_iB)$. By $1-q_i=(1+\rho_i)^{-1}$ the l.h.s.~of~(\ref{wr}) is found to be
\begin{equation*}
\Tr(\nu_\rho a_iB)=\sum_{j=1}^n(1+\rho_i)[a_ib_j]\Tr(\nu_\rho B_j)
\end{equation*}
which is the r.h.s.~by~(\ref{GaugeQFSter}). In the second case we first use (\ref{Qa}) as $\Tr(\nu_\rho a_i\str B)=q_i\Tr(\nu_\rho Ba_i\str)$ and then (\ref{wr1}). By $q_i(1-q_i)^{-1}=\rho_i$ we now find 
\begin{equation*}
\Tr(\nu_\rho a_i\str B)=-\sum_{j=1}^n\rho_i[a_i\str b_j]\Tr(\nu_\rho B_j).
\end{equation*}
We now first use the two equations for $B=a_j$ to obtain (\ref{GaugeQFSbis}) and hence (\ref{GaugeQFSter}), at which point they read like (\ref{wr}).

Finally the last sentence of the lemma follows from $\Gamma(Q)^2=\Gamma(Q^2)$, which shows that $\nu_\rho=\nu_\rho^2$ iff $\rho=0$.
\end{proof}
\begin{cor}\label{lemma:normal2} 
Any quasi-free state $\omega_P$ is realized by a density matrix (\ref{density matrix}) on a Fock space $\widetilde{\caF}$ isomorphic to $\caF$. In particular it is realized by a vector iff $P=P^2$.
\end{cor}
\begin{proof} By Lemma~\ref{lemma:diagonalization} there is a Bogoliubov transformation $\varphi$ such that $P= (\varphi\str)^{-1}\tilde P\varphi\str$ with $\tilde P$ a quasi-free state of the gauge-invariant form (\ref{GaugeQFS}). We set $\tilde A(\xi)=A((\varphi\str)^{-1}\xi)$ so that $\omega_P(\tilde A(\xi)\tilde A\str(\zeta)) = \langle \xi, S\tilde P\zeta\rangle$. The claim follows by Lemma~\ref{lemma:normal1} and in fact for the Fock space $\widetilde{\caF}$ associated with the operators $\tilde a(f)=\tilde A(f \oplus 0)$, cf.~(\ref{selfdualA}). 
\end{proof}

\begin{prop}
A quasi-free state $\omega_P$ is a squeezed state if and only if $P$ is a projection, $P=P^2$.
\end{prop}
\begin{proof}
By Corollary~\ref{lemma:normal2} $\omega_P$ is realized by a density matrix $\nu$ on $\widetilde{\caF}$. If $\omega_P$ is indecomposable then so is $\nu$, whence $\rho=0$ and $P=P^2$. Conversely, if $P$ is a projection and $\omega_P=\omega_1+\omega_2$ with general states $\omega_i\ge 0$, $\omega_i(1)>0$, then $\omega_i(A)=\omega_i(1)\omega_P(A)$, as we will show momentarily; whence $\omega_P$ is indecomposable. By the commutation relations it suffices to prove the contention for elements of the form $A=B_1\str B_2$ where $B_j$ are products of annihilation operators, not both empty. In this case both sides vanish: Clearly $\omega_P(A)=0$, because (\ref{Fock vacuum}) applies, but also $|\omega_i(B_1\str B_2)|^2\le\omega_i(B_1\str B_1)\omega_i(B_2\str B_2)$ and $\omega_i(B\str B)\le \omega_P(B\str B)=0$.
\end{proof}

We conclude that any squeezed state, as defined earlier in this section, is realized by a vector $\psi\in\caF$, since $\widetilde{\caF}=\caF$ as Hilbert spaces by the proof of Corollary~\ref{lemma:normal2}. We shall characterize $\psi$ for later use. Denoting the (non-unique) Bogoliubov transformation $\varphi$ used there by (\ref{genBog}), we have
\begin{equation*}
\tilde a(f)=A((\varphi\str)^{-1}(f \oplus 0))=a(Uf)+a\str(-\overline{V}\,\overline{f}),
\end{equation*}
where the second expression follows from (\ref{selfdualA}) in view of $(\varphi\str)^{-1} = S\varphi S$ and 
\begin{equation}
\label{inverse}
S\varphi S = 
\begin{pmatrix}
U & -\overline{V} \\ -V & \overline U
\end{pmatrix}.
\end{equation}
The squeezed state is thus characterized by $\tilde a(f)\psi = 0$, \ie
\begin{equation}\label{SqSt2}
(U\str a-V\str a\str )\psi = 0,
\end{equation}
where $a$ is shorthand for the column vector $(a_1,\ldots,a_N)$. 

We now specialize to $N=1$, thus considering a single annihilation operator $a$.
Then (\ref{SqSt2}) agrees with (\ref{SqSt}) for $\mu/\lambda= -\overline V/\overline U$, in line with (\ref{Csqst}). A general matrix $\varphi \in \mathrm{SU}(1,1)$ can be encoded by a real number $\rtn$ and a complex number $\tau = r \ep{i\theta}$ through the (unique) decomposition $\varphi=\varphi_\tau\varphi_\rtn$ with
\begin{equation*}
%\label{squeezed}
\varphi_\rtn = \begin{pmatrix} \ep{\iu\rtn} & 0 \\ 0 & \ep{-\iu\rtn} \end{pmatrix} , \qquad \varphi_\tau = \begin{pmatrix} \cosh r & \ep{\iu\theta}\sinh r \\ \ep{-\iu\theta}\sinh r & \cosh r \end{pmatrix}.
\end{equation*}
The squeezed state associated to $\varphi$ does not depend on $\rtn$ because $\varphi_\rtn$ leaves the standard Fock vacuum invariant. As we shall prove shortly, the transformation 
\begin{equation*}
a \mapsto a_\tau := (\cosh r) a + (e^{i \theta} \sinh r) a^*
\end{equation*}
is unitarily implemented on $\mathcal F$ by
\begin{equation*}
\mathcal{U}_\tau = \exp\Bigl(\frac{1}{2} \bigl(\bar{\tau} (a)^2 - \tau (a^*)^2\bigr)\Bigr),
\end{equation*}
meaning 
\begin{equation}
a_\tau = \mathcal{U}_\tau  a \mathcal{U}_\tau\str. \label{U Impl} \\
\end{equation}
In particular, $\psi = \caU_\tau \psi_0$ is the squeezed state with parameter $\tau$ since $a_\tau \psi=\caU_\tau a \psi_0 =0$, see~(\ref{SqSt}, \ref{quot}).\\

The next lemma, while standard, is stated and proved for the sake of completeness.
\begin{lemma}
\label{lemma:U squeezed}
With the notation of the previous paragraph, and $\{\vert n\rangle : n\in\bbN\}$ the usual Fock space basis of occupation numbers, we have (\ref{U Impl}) and
\begin{equation}
\mathcal{U}_\tau \ket{0} = \frac{1}{\sqrt{\cosh r}}\sum_{n\geq 0}\ep{\iu n \theta} (-\tanh r)^nq_n\vert 2n\rangle \label{Nrep}
\end{equation}
with $\tau=r\ep{\iu\theta}$ and
\begin{equation*}
q_n=\sqrt{\frac{(2n-1)!!}{(2n)!!}}=\frac{\sqrt{(2n)!}}{2^nn!}.
\end{equation*}
\end{lemma}
\begin{proof} On the one hand,
\begin{equation*}
\frac{\partial }{\partial r}a_\tau = (\sinh r) a  + (\ep{\iu\theta} \cosh r) a\str  = \ep{\iu\theta} a_\tau \str.
\end{equation*}
On the other hand,
\begin{equation*}
\frac{\partial}{\partial r}(\caU_\tau a \caU_\tau\str) = \frac{1}{2}\caU_\tau
[a, \ep{\iu\theta}(a\str)^2- \ep{-\iu\theta}a^2]\caU_\tau\str = \ep{\iu\theta}\caU_\tau a \str  \caU_\tau\str,
\end{equation*}
by $[a, (a\str)^2]=2 a\str$. Since $\caU_0a \caU_0\str  = a = a_0$, (\ref{U Impl}) follows by uniqueness of the solution of the ODE.

Eq.~(\ref{Nrep}) is similarly proved by showing that the r.h.s. satisfies the same ODE 
\begin{equation*}
\frac{dy}{dr}=\frac{1}{2}(\ep{-\iu\theta}a^2-\ep{\iu\theta}(a\str)^2)y,\qquad y(0)=1
\end{equation*}
that the l.h.s. obviously does. This follows by comparing
\begin{align}
\frac{d}{dr}\frac{(\tanh r)^n}{\sqrt{\cosh r}}%&=\frac{(\tanh r)^{n-1}}{\sqrt{\cosh r}}\Bigl(\frac{n}{\cosh^2r}-\frac{1}{2}\tanh^2 r\Bigr)\nonumber\\
&=\frac{(\tanh r)^n}{\sqrt{\cosh r}}\Bigl(n(\tanh r)^{-1}-\bigl(n+\frac{1}{2}\bigr)\tanh r\Bigr)
\label{bra}
\end{align}
with
\begin{gather*}
a^2\vert 2n\rangle =\sqrt{2n(2n-1)}\vert 2(n-1)\rangle,\\
(a\str)^2\vert 2n\rangle =\sqrt{(2n+2)(2n+1)}\vert 2(n+1)\rangle,\\
q_n\sqrt{2n(2n-1)}=q_{n-1}(2n-1)=2q_{n-1}\bigl((n-1)+\frac{1}{2}\bigr),\\
q_n\sqrt{(2n+2)(2n+1)}=q_{n+1}(2n+2)=2q_{n+1}(n+1),
\end{gather*}
where we used the first expression for $q_n$. The action of $\ep{-\iu\theta}a^2/2$ on the r.h.s. of~(\ref{Nrep}), followed by the replacement of $n-1$ with $n$ in the sum, leads to the second term in the bracket (\ref{bra}), in line with the claim. The action of $-\ep{\iu\theta}(a\str)^2$ is likewise seen to match the first term by means of the opposite shift of $n$.

The other writing of $q_n$ follows by $(2n)!! = 2^n n!$ and $(2n-1)!! = (2n)! / 2^n n!$.
\end{proof}

The transition probabilities seen in Theorem~\ref{THM} follow by inserting $\tanh r$ $=1/\sqrt 2$, $\cosh r=\sqrt 2$ in (\ref{Nrep}). 

%%%%%%%%%%%%%%%%%
%%%%%geometry
%%%%%%%%%%%%%%%%%
\subsection{Geometry of squeezed states}
\label{sec:geom}
We start by recalling general geometric facts from adiabatic theory, see e.g.~\cite{AFG}. Given a family of projections $P_\varphi$ depending on some parameters $\varphi$, the adiabatic connection is given by
\begin{equation*}
%\label{connection}
\mathcal{A} = P_\perp \d P, \qquad P_\perp = 1-P.
\end{equation*}

A Hermitian structure is defined on the tangent space at $\varphi$ by 
\begin{equation*}
h := 2 \tr(\mathcal{A} \otimes \mathcal{A}^*)
\end{equation*}
The symmetric and antisymmetric parts of $h$ define the Fubini-Study metric $g$ and the adiabatic curvature $\omega$, namely
\begin{equation*}
h=:g - \iu \omega.
\end{equation*}
Note that $\otimes$ above and in the rest of this section refers to the tensor algebra generated by $d\varphi$, not to the Hilbert space tensor product. The parameter space is thereby endowed with a natural geometric structure. 

We shall describe the geometric structure associated to squeezed states in holomorphic coordinates. We recall that $\varphi^*$ and $S\varphi S$ are inverses, which by~(\ref{stared}, \ref{inverse}) is equivalent to either line of
\begin{align}
U^* U - V^* V &=1,  \qquad U^* \overline{V} = V^* \overline{U}, \label{firstline}\\
U U^* - \overline{V} V^t  &=1,  \qquad UV^* = \overline{V} U^t. \nonumber %\label{secondline}
\end{align}
The normalized squeezed state $\psi$ associated with the Bogoliubov transformation $\varphi$ satisfies (\ref{SqSt2}). Since $U$ is invertible in view of the first relation in (\ref{firstline}), an equivalent set of equations is
\begin{equation}\label{zformula}
(a + Z a\str)\psi = 0,\qquad Z:= -(U\str)^{-1} V\str,
\end{equation}
where $a=(a_1,\ldots,a_N)$ as before. The two relations (\ref{firstline}) respectively imply, after inverting the second one,
\begin{equation}
\label{zzformula}
Z Z\str = 1 - (U U\str)^{-1} < 1,\qquad Z=Z^t.
\end{equation}
Conversely, if $Z$ satisfies these two properties, then it can be written as in (\ref{zformula}) with $U$, $V$ unique up to multiplication by a unitary from the right. By introducing the bilinear expression $(a,b)=\sum_ia_ib_i$, we have $(a\str,Z a\str)= \sum_{ij} Z_{ij}  a_i^* a_j^*$ and
\begin{equation*}
\ep{-\frac{1}{2} (a\str,Z a\str)}a\ep{\frac{1}{2} (a\str,Z a\str)} =  a + Z a\str.
\end{equation*}
Hence
\begin{equation*}
(a+Z a\str)\ep{-\frac{1}{2} (a\str,Z a\str)}\psi_0 = \ep{-\frac{1}{2} (a\str,Z a\str)}a\psi_0 = 0,
\end{equation*}
where $\psi_0$ is the Fock vacuum, so that $\ep{-\frac{1}{2} (a\str,Z a\str)}\psi_0$ is proportional to $\psi$. The normalization constant is given by
\begin{equation}\label{psiNormalization}
\bigl\Vert \ep{-\frac{1}{2} (a\str,Z a\str)} \psi_0\bigr\Vert^2 = \det(1-Z Z\str)^{-1/2},
\end{equation}
whence
\begin{equation*}
\psi = \det(1-Z Z\str)^{1/4}\ep{-\frac{1}{2} (a\str,Z a\str)}\psi_0.
\end{equation*}
The norm~(\ref{psiNormalization}) can be computed using Takagi's factorization $Z= W^t  D W$, where $W$ is unitary and $D$ is the diagonal matrix of singular values. With this,
\begin{equation*}
(a\str, Z a\str) = (b\str, D b\str),\qquad b\str = W a\str
\end{equation*}
is a Bogoliubov transformation preserving the vacuum state. Hence,
\begin{align*}
\bigl\Vert \ep{-\frac{1}{2} (a\str,Z a\str)} \psi_0\bigr\Vert^2 &= \prod_{i=1}^N\langle \ep{-\frac{1}{2}D_i b_i\str b_i\str}\psi_0,\ep{-\frac{1}{2}D_i b_i\str b_i\str}\psi_0\rangle \\
&=\prod_{i=1}^N\sum_{k=0}^\infty \frac{D_i^{2k}}{2^{2k}(k!)^2}\langle \psi_0,b_i^{2k} (b_i\str)^{2k}\psi_0\rangle \\
&=\prod_{i=1}^N\sum_{k=0}^\infty \frac{D_i^{2k}(2k)!}{2^{2k}(k!)^2} = \prod_{i=1}^N(1-D_i^2)^{-1/2},
\end{align*}
which is~(\ref{psiNormalization}). 

Summarising, the set of squeezed states corresponds to the manifold 
\begin{equation*}
\mathcal{M} = \{Z \in \mathrm{GL}(N, \mathbb{C}) : Z = Z^t, ZZ\str < 1\}.
\end{equation*}
This is an open subset of $\mathbb{C}^{N(N+1)/2} \simeq \{Z_{ij},\,1\leq i \leq j \leq N \}$ and inherits the natural complex structure $J$ and Dolbeault decomposition $d = \partial + \bar \partial$. We write $\d$ instead of $d$ if it just acts on the first factor to its right. We denote by $\d Z$ (resp.~$\d Z^*$) a matrix valued 1-form with entries $\d Z_{ij}$ (resp.~$\d\overline{{Z}_{ji}}$). Moreover $\mathcal{M}$ carries the family of projections $P=\ket{\psi} \bra{\psi}$.

\begin{thm}
\label{Ngeom} The family $P$ turns $\mathcal{M}$ into a K\"{a}hler manifold
$(\mathcal{M},\, g,\, \omega)$ with Hermitian structure
\begin{equation}\label{hstrct}
h =\tr \bigl((1-ZZ\str)^{-1}\d Z \otimes(1-Z\str Z)^{-1} \d Z\str\bigr).
\end{equation}
It is also expressed as
$$
h = \tr \bigl((U\str \d Z \overline U) \otimes (U\str \d Z \overline U)^* \bigr),
$$
where $U$ is determined by $Z$ as in (\ref{zzformula}) and to sufficient extent so as to make the r.h.s. well-defined.
\end{thm}

\begin{proof}
The map
$$
Z \mapsto \ket{\tilde\psi} := \ep{-\frac{1}{2} (a\str,Z a\str)}\psi_0,
$$
is holomorphic, but not so after normalization, $\ket{\psi} = \No^{-1} \ket{\tilde{\psi}}$. Nonetheless we have 
$$P_\perp \bar \partial\ket{\psi} =0,\qquad 
P_\perp \partial\ket{\psi} =\No^{-1} P_\perp \partial\ket{\tilde\psi},$$
where the first equation implies $P_\perp\bar \partial( \ket{ \psi}\bra{ \psi})=0$ and thus 
\begin{equation*}
%\label{hol1}
\mathcal{A} = P_\perp \partial P=P_\perp (\partial\ket{\psi})\bra{\psi}.
\end{equation*}
With the help of the defining relation $J \partial = \iu \partial$ we then get $\mathcal{A}(J X) = \iu \mathcal{A}(X)$ and $h(JX,X') = \iu  h(X,X')$ for any tangent vectors $X,X'$. The latter is equivalent to the consistency condition $g(JX,X')  = \omega(X,X')$ required for a K\"{a}hler manifold. The remaining condition that $\omega$ is closed, \ie $\d \omega =0$, is always satisfied by the adiabatic curvature.

We claim that 
$$
 h = 2 \partial \otimes \bar{\partial} \log \No^2, 
$$
which implies that $2\log \No^2$ is the K\"{a}hler potential, \ie 
$\omega=\iu\partial \bar{\partial} \log \No^2$.

This follows by
\begin{align*}
  \tr(\mathcal{A} \otimes \mathcal{A}^*) &=\tr(P_\perp \partial\ket{\psi}\otimes \bar \partial\bra{\psi})
 %\frac{1}{\No^4}\tr \left((1 - \ket{\psi} \bra{\psi}) \partial \ket{\tilde{\psi}} \bra{\tilde{\psi}} \otimes \ket{\tilde{\psi}} \bar{\partial}\bra{\tilde{\psi}}\right)
\\
 &= \frac{1}{\No^4} \bigl( \No^2 \tr(\partial \ket{\tilde{\psi}} \otimes \bar{\partial} \bra{\tilde{\psi}} ) - \partial \langle{\tilde{\psi}} | \tilde{\psi} \rangle \otimes \bar{\partial} \langle \tilde{\psi} | \tilde{\psi} \rangle \bigr) \\
   &=\frac{1}{\No^4} \bigl( \No^2 \partial \otimes \bar{\partial} \No^2 - \partial \No^2 \otimes \bar{\partial} \No^2 \bigr) 
   = \partial \otimes \bar{\partial} \log \No^2.
\end{align*}
With (\ref{psiNormalization}) and
$\log\No^2 = \log\det (1- Z Z\str)^{-1/2} = -\frac{1}{2} \tr \log (1- Z Z\str)$ 
we get
$$
h = -\partial \otimes \bar{\partial} \tr \log (1- ZZ\str).
$$
Differentiating under the trace and making use of $\d {A}^{-1} = -A^{-1} \d A A^{-1}$, we find
$$
h = \tr\bigl((1-ZZ\str)^{-1} \d Z \otimes \d Z^* + (1-ZZ\str)^{-1} \d Z Z^* \otimes(1-ZZ\str)^{-1} Z \d Z^* \bigr).
$$
Eq.~(\ref{hstrct}) follows by regrouping terms using the identity 
$$1+Z\str (1-ZZ\str)^{-1} Z = (1-Z\str Z)^{-1},$$
which is seen from $Z(1-Z\str Z)=(1-ZZ\str)Z$. The alternate
expression then follows from (\ref{zzformula}), also by way of $Z\str
Z=(ZZ\str)^t$. We observe that it is not affected when $U$ is
multiplied by a unitary from the right, as allowed by $Z$.
%Now $g$ and $\omega$ follow by taking the symmetric resp. the
%antisymmetric part of the expression for $h$.
\end{proof}

\begin{proof}[Proof of Proposition~\ref{sqpr}] In the case $N=1$ we have 
$h= (1 -|z|^2)^{-2}\d z\otimes \d \bar{z}$
and the result follows by decomposing $h=g - \iu \omega$.
\end{proof}
%%%%%%%%%
%%%%%%%%%

\section{The classical equations of motion}
\label{sec:flow}

As before, we set $\rate = 1$. The canonical equations~(\ref{ClassEq}) are equivalent to Newton's equation
\begin{equation}\label{Newton}
\ddot x(t) + t^2 x(t) = 0,
\end{equation}
which takes the form of a special case of the Weber differential equation. The two linearly independent solutions may be chosen even and odd in $t$,
\begin{equation*}
x_\pm (-t) = \pm x_\pm(t).
\end{equation*}
They are unique up to multiples. Imposing a normalization on the Wronskian (which is constant in $t$),
\begin{equation}
\label{Wronskian}
W(x_+,x_-) := x_+\dot x_- - \dot x_+ x_-= 1,
\end{equation}
leaves one free parameter. The general (real) solution of~(\ref{Newton}) is

\begin{equation*}
x(t) = \rtn_+ x_+(t) + \rtn_- x_-(t)
\end{equation*}
with arbitrary coefficients $\rtn_\pm\in\bbR$, which are in turn determined by the solution itself:
\begin{equation}\label{rtnW}
\rtn_- = W(x_+,x),\qquad \rtn_+ = -W(x_-,x).
\end{equation}

We shall derive the propagator for some time interval~$[t_1,t_2]$ as a map on phase space~$\bbR^2\ni(x,p)$, and express it in the complex coordinate $a$ 
\begin{equation}\label{xpa}
x = \frac{1}{\sqrt{2\omega}}(\bar a + a), \qquad p= \iu \sqrt{\frac{\omega}{2}}(\bar a - a),
\end{equation}
for some fixed $\omega>0$, as defined in~(\ref{Classical a}) and for reasons explained there. The solutions $x_\pm(t)$ give rise to
\begin{equation}
\label{eq:18}
a_\pm(t) = \frac{1}{\sqrt{2\omega}}(\omega x_\pm(t) +\iu \dot x_\pm(t))
\end{equation}
with
\begin{equation}\label{aSymm}
a_\pm(-t) = \pm \overline{a_\pm(t)}.
\end{equation}
The general solution
\begin{equation}\label{GenSol a}
a(t) = \rtn_+ a_+(t) + \rtn_- a_-(t),
\end{equation}
being complex, determines both real amplitudes $\rtn_\pm$, without resorting to derivatives as in~(\ref{rtnW}): Expressing there $x$ and $\dot x = p$ by~(\ref{xpa}), we obtain
\begin{equation}\label{rtns}
\begin{aligned}
\rtn_+ &= \iu \overline{a_-(t)} a(t) + \mathrm{c.c.},\\
\rtn_- &= -\iu \overline{a_+(t)} a(t) + \mathrm{c.c.}.
\end{aligned}
\end{equation}
Using~(\ref{GenSol a}) at $t = t_2$ and~(\ref{rtns}) at $t = t_1$, we get
\begin{multline*}
a(t_2) =  \bigl(\iu a_+(t_2) \overline{a_-(t_1)} - \iu a_-(t_2) \overline{a_+(t_1)}\bigr) a(t_1)\\ + \bigl(- \iu a_+(t_2) {a_-(t_1)} + \iu a_-(t_2) {a_+(t_1)}\bigr) \overline{a(t_1)}.
\end{multline*}
The propagator is thus of the form~(\ref{ClassicalFlow}) with
\begin{equation}\label{UV}
\begin{aligned}
U &= \iu a_+(t_2) \overline{a_-(t_1)} - \iu a_-(t_2) \overline{a_+(t_1)}, \\
\overline{V} &= - \iu a_+(t_2) {a_-(t_1)} + \iu a_-(t_2) {a_+(t_1)}.
\end{aligned}
\end{equation}
In particular, for $t_1 = -t$, $t_2 = t$, we have by~(\ref{aSymm})
\begin{equation}\label{eq:22}
U = -2 \iu a_+(t) a_-(t), \qquad
\overline{V} = 2\iu\Re\bigl( \overline{a_+(t)} {a_-(t)}\bigr). 
\end{equation}
According to~\cite[19.1.5 and 19.2.1]{AS}, the even and odd solutions of~(\ref{Newton}) are the parabolic cylinder functions
$$
\ep{-\iu t^2/2} M\bigl(\frac{1}{4}, \frac{1}{2}, \iu t^2\bigr), \qquad t \ep{-\iu t^2/2} M\bigl(\frac{3}{4}, \frac{3}{2}, \iu t^2\bigr)
$$
expressed in terms of the confluent hypergeometric function $M(a,b,z)$. Using its differentiable asymptotics for $z \to \infty$ \cite[13.5.1]{AS}, we find for $t\to+\infty$
\begin{gather*}
x_\pm(t) = x_\pm t^{-1/2} (\cos\theta_\pm(t)+ o(1)), \\
\theta_+(t) = \frac{t^2}{2} -\frac{\pi}{8}, \qquad \theta_-(t) = \frac{t^2}{2} - \frac{3 \pi}{8}
\end{gather*}
with arbitrary amplitudes $x_\pm$, as well as similar expressions for $\dot{x}_\pm(t)$ with $t^{-1/2}$ replaced by $t^{1/2}$ and $\cos$ by $-\sin$. We observe that 
\begin{equation*}
\theta_+(t) -\theta_-(t) =\frac{\pi}{4}, \qquad
\theta_+(t) +\theta_-(t) =t^2 - \frac{\pi}{2},
\end{equation*}
and so obtain from~(\ref{Wronskian})
\begin{align*}
W(x_+, x_-) &= -x_+ x_- (\cos\theta_+\sin\theta_- - \cos\theta_- \sin\theta_+ )\\		   
&= x_+ x_- \sin(\theta_+ - \theta_-) = \frac{x_+ x_-}{\sqrt{2}},
\end{align*}
implying $x_+ x_- = \sqrt{2}$. We also obtain from~(\ref{eq:18}) with $\omega = t$,
\begin{equation}\label{apm}
a_\pm(t) = \frac{x_\pm}{\sqrt{2}} \ep{-\iu \theta_\pm(t)} + o(1), \qquad (t \to \infty).
\end{equation}
In particular~(\ref{eq:22}) becomes
\begin{align*}
U &= - \sqrt{2} \iu \ep{-\iu (t^2 - \frac{\pi}{2})}+o(1), \nonumber \\
\overline{V} &= 2\iu \Re( \frac{1}{\sqrt{2}}\ep{\iu(\theta_+ - \theta_-)}) + o(1) = \iu + o(1),
\end{align*}
which are the asymptotic expressions used in the proof of Theorem~\ref{THM}.

%%%%%%%%%%%%%%%%%%%%%%%%%%%%%%%%%%

\section{Extensions}\label{sec:Extensions}
Two extensions of the above problem are considered.
\subsection{Asymmetric time intervals}
In a slight generalization of Theorem~\ref{THM} we consider asymmetric initial and final times, $t_1\to-\infty$ and $t_2\to+\infty$, together with initial (ground) and final (squeezed) states now understood with respect to different Hamiltonians, $H_{t_1} \neq H_{t_2}$. The conclusions are unchanged, up to the angle now being
\begin{equation*}
\theta = -\rate t_2^2 - \frac{\pi}{2}+ o(1).
\end{equation*}
This is in line with the remark made earlier by which the squeezed state is formed at the spectral collapse and only the time $t_2$ elapsed since then contributes to its rotation. 

The argument is as follows ($\rate=1$). We introduce coordinates $a_i$, $\overline{a_i}$, ($i=1,2$) like in (\ref{Classical a}) but with different frequencies $\omega_i$. We so associate $a_i(t)$ to any solution $x(t)$, and in particular $a_{i\pm }(t)$ to $x_\pm(t)$, see (\ref{eq:18}). Then (\ref{GenSol a}), now decorated with $i=1,2$, still holds true, and in fact with common coefficients $\rtn_\pm$, since those are determined by $x(t)$ without reference to $\omega_i$, cf.~(\ref{rtnW}). As a result the propagator between times $t_1$ and $t_2$ relating $a_1(t_1)$ to $a_2(t_2)$ (with conjugates) is the Bogoliubov transformation
\begin{align*}
U &= \iu a_{2+}(t_2) \overline{a_{1-}(t_1)} - \iu a_{2-}(t_2) \overline{a_{1+}(t_1)}, \\
\overline{V} &= - \iu a_{2+}(t_2) {a_{1-}(t_1)} + \iu a_{2-}(t_2) {a_{1+}(t_1)},
\end{align*}
cf. (\ref{UV}). We then pick $\omega_i=|t_i|$ as understood in the claim, but unlike there we first let both $t_1,t_2\to +\infty$. Then (\ref{apm}) together with
\begin{equation*}
\theta_+(t_1) -\theta_-(t_2) =\frac{t_1^2-t_2^2}{2}+\frac{\pi}{4}, \qquad
\theta_-(t_1) -\theta_+(t_2) =\frac{t_1^2-t_2^2}{2}-\frac{\pi}{4}
\end{equation*}
lead to 
\begin{equation}\label{ClassicalSol_TwoTimes}
U=\ep{\iu (t_1^2-t_2^2)/2}+ o(1),\qquad \overline{V}=o(1).
\end{equation}
We observe that, up to $o(1)$, this transformation does not squeeze the vacuum and rotates squeezed states by the angle $t_1^2-t_2^2$. Its diagonal form is in line with the adiabatic theorem valid away from the collapse: The Heisenberg evolution respects the instantaneous creation and annihilation operators and the Schr\"odinger evolution respects the instantaneous eigenstates.

We finally return to $t_1\to-\infty$. We split the interval $[t_1,t_2]$ at the point $|t_1|$, possibly lying outside of it. The combined transformation of~(\ref{ClassicalSol}) for $t=\vert t_1\vert$ followed by~(\ref{ClassicalSol_TwoTimes}) is given by
\begin{equation*}
 U = -\sqrt 2 \iu \ep{-\iu(t_1^2+t_2^2+\frac{\pi}{2})/2}+o(1),\qquad\overline{ V}=\iu\ep{\iu(t_1^2-t_2^2)/2}+o(1),
\end{equation*}
which yields a joint squeezing ratio of
\begin{equation*}
\frac{\mu}{\lambda} = -\frac{\overline{ V}}{\overline{ U}} = -\frac{1}{\sqrt 2}\ep{-\iu(t_2^2+\frac{\pi}{2})} + o(1).
\end{equation*}

\noindent\textbf{Remark.} One could further consider the case $t_1\to-\infty$ while $t_2$ remains finite. This was represented in Figure~\ref{fig:SqueezingPlot}, where we show the pattern traced by the solution in the Poincar\'{e} disk relative to the instantaneous ground state, and plot its squeezing $\tanh r$ as a function of $t_2$. 

\subsection{Gapped case}

A second extension is concerned with the gapped case, such as $\omega_t^2 = \rate^2 t^2 + g^2$. The scaling argument of Lemma~\ref{self} shows that the squeezing and the phase slip of the future asymptotic state depend only on the ratio $\delta^2 := g^2/\rate$.  For fixed $g>0$, the result then trivializes in the adiabatic limit $\rate \to 0$, in the sense that instantaneous ground states are respected by the dynamics. A non-trivial result is obtained in the cross-over regime given by the scaling $g = \sqrt \rate \delta$ with fixed $\delta$.

Without loss, we shall consider a symmetric time interval $[-t,t]$ and forgo for simplicity the phase of the squeezing.
\begin{prop}
\label{PROP}
In the situation of Theorem~\ref{THM}, but with $\omega_t^2 = \rate^2 t^2 + g^2$, we have
\begin{equation*}
\tanh r  = \frac{1}{\sqrt{1+\ep{\pi\delta^2}}}, \qquad \delta^2 = \frac{g^2}{\rate},
\end{equation*}
\ie $r = \log\bigl( \sqrt{1+\ep{-\pi \delta^2}} + \sqrt{\ep{-\pi \delta^2}}\bigr)$ in the limit of an infinite time interval.
\end{prop}
\noindent\textbf{Remark.} The result monotonically interpolates
between $\tanh r =1/ \sqrt 2$ for $\delta = 0$ and $r=0$ for
$\delta\to\infty$. \\

\noindent\textbf{Remark.} The fidelity is 
\begin{equation*}
p_0 = \frac{1}{\sqrt{1 + \ep{-\pi g^2/\alpha}}}, 
\end{equation*}
as seen from $p_0=\cosh r = (1-\tanh^2r)^{-1/2}$. For $g^2/\alpha\gg 1$ it
approaches $1$ and the tunneling is asymptotically given by
\begin{equation*}
1-p_0\approx \frac{1}{2}\ep{-\pi g^2/\alpha}.
\end{equation*}
It should be compared with the Landau-Zener formula for an avoided
crossing sharing the same gap $\sqrt{\rate^2 t^2 + g^2}$, which reads
$1-p_0=\ep{-\pi g^2/2\alpha}$. \\

\begin{proof} By scaling $s = \sqrt\rate t$, we reduce matters to $\rate = 1$. Weber's equation~(\ref{Newton}) now reads
\begin{equation*}
\ddot x(t) + (t^2 + \delta^2)x(t) = 0.
\end{equation*}
Using the same references as before, the even and odd solutions are 
\begin{equation*}
\ep{-\iu t^2/2} M\bigl(\frac{1}{4}(1+\iu\delta^2), \frac{1}{2}, \iu t^2\bigr), \qquad t \ep{-\iu t^2/2} M\bigl(\frac{1}{4}(3+\iu\delta^2), \frac{3}{2}, \iu t^2\bigr)
\end{equation*}
and have asymptotics
\begin{equation*}
x_\pm(t) = x_\pm t^{-1/2}\bigl(\gamma_\pm \ep{\iu\theta_\pm(t)} + \mathrm{c.c.} + o(1)\bigr),
\end{equation*}
as $t\to\infty$, where
\begin{equation*}
\gamma_+ = \Gamma\bigl((1+\iu\delta^2)/4\bigr)^{-1},\qquad \gamma_- = \Gamma\bigl((3+\iu\delta^2)/4\bigr)^{-1},
\end{equation*}
and
\begin{equation*}
\theta_+(t) = \frac{t^2}{2} +\frac{\delta^2}{2}\log t- \frac{\pi}{8}, \qquad \theta_-(t) = \frac{t^2}{2} +\frac{\delta^2}{2}\log t - \frac{3 \pi}{8}.
\end{equation*}
The real constants $x_\pm$, which depend on $\delta$, are supposed to obey the normalization condition $W(x_+,x_-) = 1$. For $\delta = 0$, we recover the previous expressions though the same normalization of $x_\pm$ is recovered only after replacing $\gamma_\pm$ by $1/2$. We also observe that the logarithmic correction reflects the WKB approximation, in that
\begin{equation*}
\theta_\pm'(t) = t + \frac{\delta^2}{2t} = \sqrt{t^2 + \delta^2} + o(1).
\end{equation*}
For $\omega_t = t + o(1)$, we find
\begin{equation*}
a_\pm(t) = \sqrt 2 x_\pm \overline{\gamma_\pm}\ep{-\iu\theta_\pm(t)} + o(1)
\end{equation*}
by~(\ref{eq:18}). Now, (\ref{eq:22}) yields
\begin{align*}
U &= -4\iu x_+x_-\overline{\gamma_+}\overline{\gamma_-}\ep{-\iu(\theta_+(t)+\theta_-(t))} + o(1),\\
\overline V &= 4\iu x_+x_- \mathrm{Re} \bigl(\gamma_+ \overline{\gamma_-}\ep{\iu\pi / 4} \bigr) + o(1)
\end{align*}
and
\begin{equation}\label{deltasqueezing}
\ep{\iu\theta}\tanh r = -\frac{\mathrm{Re} \bigl(\gamma_+ \overline{\gamma_-}\ep{\iu\pi / 4} \bigr)}{{\gamma_+}{\gamma_-}}\ep{-\iu(\theta_+(t)+\theta_-(t))} + o(1)
\end{equation}
by~(\ref{Csqst}). We then use Euler's reflection formula $\Gamma(1-z)\Gamma(z) =\pi/\sin(\pi z)$ and $ \overline{\Gamma(z)}=\Gamma(\bar z)$ to conclude that
\begin{equation*}
\gamma_+ \overline{\gamma_-} = \frac{1}{\pi}\sin\frac{\pi}{4}(1+\iu\delta^2) = \frac{1}{2\pi\iu}\bigl(\ep{\iu\pi /4}\ep{-\pi\delta^2 /4} - \ep{-\iu\pi /4}\ep{\pi\delta^2 /4}\bigr).
\end{equation*}
Thus, in the limit $t\to\infty$,
\begin{equation*}
\tanh r = \frac{\bigl\vert \mathrm{Re} (\gamma_+ \overline{\gamma_-}\ep{\iu\pi / 4})\bigr\vert }{\bigl\vert {\gamma_+}\overline{\gamma_-}\bigr\vert}
= \frac{\ep{-\pi\delta^2/4}}{\sqrt{\ep{\pi\delta^2/2} + \ep{-\pi\delta^2/2}}} = \frac{1}{\sqrt{1+\ep{\pi\delta^2}}}.
\end{equation*}
\end{proof}

\noindent\textbf{Remark.} The phase of the squeezing $\theta$ can be obtained from~(\ref{deltasqueezing}) as
\begin{equation*}
\theta = - t^2 - \delta^2 \log t -\frac{\pi}{2} - \mathrm{arg}({\gamma_+}{\gamma_-})  + o(1),
\end{equation*}
where $\mathrm{arg}({\gamma_+}{\gamma_-}) $ is independent of $t$, and $\mathrm{arg}({\gamma_+}{\gamma_-}) \to 0$ as $\delta\to 0$.

%%%%%%%%%%%%%%%%%%%%%%%%%%%%%%%%%%

\section{Summary}\label{sec:Conclusions}

We considered a time-dependent quadratic Hamiltonian
$$
H_t =\frac{1}{2} (p^2 + \omega_t^2 x^2), \qquad \omega^2_t = \alpha^2 t^2 + g^2.
$$
We derived the solution of the associated driven Schr\"{o}dinger equation that is initiated at the instantaneous ground state in the distant past. The time evolved wave-function is squeezed upon crossing of the non-adiabatic region around $t=0$ and we determine its squeezing parameters asymptotically as $t \to \infty$; see Theorem~\ref{THM} for the case of non-avoided crossing $g=0$ and Proposition~\ref{PROP} for the case of avoided crossing $g>0$. In particular the probability $p_0$ to find the state in the instantaneous ground state (fidelity with respect to the ground state) as $t \to \infty$ is given by
$$
p_0 = \frac{1}{\sqrt{1 + \ep{-\pi g^2/\alpha}}}.
$$
For a non-zero $g$ and $\alpha \to 0$ the fidelity exponentially
approaches $1$, and hence the tunneling is given by a Landau-Zener
type formula. For $g = 0$ the fidelity is constant
and equal to $1/ \sqrt{2}$. This manifests breaking of the adiabatic
theory at points where infinitely many eigenvalues collide.\\

\noindent\textbf{Acknowledgments.} We thank T. Esslinger and his
group for discussions which led us to study this problem.  

\bibliography{oscillator}
\bibliographystyle{unsrt}

%%%%%%%%%%%%%%%%%%%%%%%%%%%%%%%%%%%%%%%%%%%%%%%%%%%%%%%%%%%%%%%%%%%%%%%%

%
\end{document}